\documentclass[twocolumn]{autart}
\usepackage{graphicx}
\usepackage{cite}
\usepackage{amsmath,amssymb,amsfonts,mathtools}
\usepackage{textcomp}
\usepackage{booktabs}
\usepackage{array}
\usepackage{xcolor}
\usepackage{algorithm}
\usepackage{algpseudocode}

\newtheorem{proposition}{Proposition}
\newtheorem{lemma}{Lemma}

\newtheorem{remark}{Remark}
\newenvironment{proof}{{\bfseries Proof.}}{\qed}

\newcommand{\MB}[1]{{\color{black}#1}}

\begin{document}

\begin{frontmatter}

\title{Rao--Blackwellized Stein Gradient Descent for Joint State--Parameter Estimation}

\author[idsia]{Milad~Banitalebi~Dehkordi}\ead{milad.banitalebi@supsi.ch},
\author[idsia]{Manas~Mejari}\ead{manas.mejari@supsi.ch},
\author[idsia]{Dario~Piga}\ead{dario.piga@supsi.ch}
\address[idsia]{IDSIA Dalle Molle Institute for Artificial Intelligence, SUPSI-USI, 6962 Lugano, Switzerland}

\begin{abstract}
We present a filtering framework for online joint state estimation and parameter identification in nonlinear, time-varying systems. The algorithm uses a \emph{Rao–Blackwellization} technique to infer joint state–parameter posteriors efficiently. In particular, conditional state distributions are computed analytically via Kalman filtering, while model parameters, including the measurement-noise covariance, are approximated using particle-based \emph{Stein Variational Gradient Descent} (SVGD), enabling stable real-time inference. To handle parameters subject to physical constraints, we further introduce constrained variants that enforce them through an alternating direction method of multipliers (ADMM) splitting of the SVGD update, including nonlinear equality constraints that standard particle filters cannot readily handle.
We derive a stability bound that relates the approximation error in the parameter posterior to the resulting error in the marginal state distribution. 
Performance of the proposed filters is validated on three case studies: a fed-batch bioreactor with Haldane kinetics and a damped pendulum, both under physical constraints, and a neural-network-augmented dynamic system. The examples cover parameter estimation under inequality and equality constraints and online neural-network training within a dynamical model. 
\end{abstract}

\end{frontmatter}

\section{Introduction}

Accurate  state estimation  requires a well-specified process model and known measurement noise statistics. In practice, however, model parameters and noise covariances are typically  unknown and may change over time due to changing operating conditions. Such uncertainties can degrade the performance of the classical filters that assume fixed and known models~\cite{huang2017novel}. This motivates joint state and parameter estimation with uncertainty quantification. This motivates methods that infer the joint posterior over both states and unknown parameters. 

Traditional approaches for joint state and parameter estimation often rely on \emph{augmented} Extended Kalman Filters (EKF), where the parameters are appended to the state vector and updated with the system states. Although augmented EKFs have been widely used \cite{van2001square}, they are known to perform poorly for highly nonlinear systems.  
Moreover, learning additional quantities such as process and measurement noise statistics is numerically delicate and may lead to instability or even filter divergence \cite{ge2016performance}.  Several adaptive Kalman filtering methods have been proposed to estimate noise covariances (see, e.g,~\cite{odelson2006new,karasalo2011optimization,ge2024adaptive}), but they typically suffer from sensitivity to model mismatch and tuning issues. Beyond joint state–parameter estimation, filtering-based techniques have also been applied to problems concerned solely with parameter learning. For instance, \cite{bemporad2022rnnkalman} reformulates recurrent neural-network training as a state-estimation problem and applies the Extended Kalman Filter (EKF) to obtain second-order-like updates, while \cite{piga2020rao} introduces a particle-filtering method for recursive Bayesian inference of Piecewise Affine  models.

Variational Inference (VI) methods can also approximate the joint posterior distribution over states and parameters \cite{archer2015black}, but they are not typically   suited for real-time  applications. Monte Carlo methods, such as Particle Filters (PFs), offer a more flexible framework for approximating complex, potentially non-unimodal  distributions \cite{hartikainen2013variational, arulampalam2002tutorial}. However, a na\"ive application of PFs to jointly propagate both states and parameters is restrictive for high dimensional systems and often fails to exploit the known structure of the state-space model. Rao--Blackwellized Particle Filters (RBPFs) address this limitation by analytically evaluating the conditional distribution of the states given the parameters and using a Monte Carlo approximation for the parameter distribution \cite{murphy2001rao, mejari2022maximum, Schon, badar2024rao}. This approach leverages the prior knowledge encoded in the state-space model while retaining the flexibility of particle-based methods for parameter inference. 
In standard RBPFs, the temporal evolution of parameters is usually modeled as a stochastic process with additive noise, which is chosen as \emph{proposal distribution} from which the particles are sampled. This  introduces an additional hyperparameter, namely the noise covariance for parameter propagation. This  often requires careful tuning and can affect the accuracy of both state and parameter estimates. 

Recently, Stein Variational Gradient Descent (SVGD) has emerged as a promising deterministic alternative for approximating probability distributions \cite{liu2016stein}. Unlike PFs, which propagate  particles stochastically, weighted based on likelihoods, SVGD evolves particles deterministically along directions that minimize the Kullback--Leibler (KL) divergence between the approximate and true posterior. This  makes SVGD appealing in a Rao--Blackwellized framework, where the states can be estimated analytically (e.g., via the EKF) and the parameter distribution is approximated deterministically using SVGD. Building on this, in this paper we introduce the \emph{Rao--Blackwellized Stein Gradient Descent (RBSGD) filter}, where the measurements  deterministically steer the particles toward the true parameter distribution. This  allows for accurate state estimation, avoiding the need to carefully tune proposal parameters.  To the best of our knowledge, the integration of Rao-Blackwellization concepts and SVGD for joint state-parameter estimation has not been proposed before. 

In addition, curvature information in the parameter space has been incorporated in variants of SVGD, such as Newton SVGD or L-BFGS SVGD, to improve convergence by adapting the step size according to the local geometry of the likelihood \cite{liu2018riemannian,detommaso2018stein, maken2022stein}. While effective, these approaches typically rely on computing the Hessian of the log-posterior, which introduces significant computational overhead. To address this, in this paper we also propose the \emph{Rao--Blackwellized Fisher Stein Gradient Descent (RBFSGD) filter}, an extension of RBSGD that incorporates the Fisher Information Matrix (FIM) as an indicator of the local parameter space geometry \cite{amari1998natural, hwang2024fadam, gomes2025adafisher}. The FIM is utilized in two ways: (i) in updating the parameters using a Fisher-adapted Adam scheme, and (ii) in the kernel formulation of SVGD, where the standard Euclidean distance is replaced by a Fisher-informed distance in the RBF kernel. This approach improves the scaling and directionality of parameter updates while maintaining computational tractability.

\MB{We further extend both filters to constrained parameter estimation through an ADMM
splitting of the SVGD update, yielding the constrained variants cRBSGD and cRBFSGD. This
enforces box, affine, and nonlinear equality constraints on the parameters in a unified way,
including constraints such as nonlinear equalities that standard particle filters cannot
readily impose, since projecting resampled particles onto a nonlinear constraint set distorts
their importance weights.} Finally, we analyze the theoretical properties of the proposed RBSGD and RBFSGD filters. In particular, we establish an upper bound on the $\ell_1$-norm error between the approximate and true marginal state distributions. Furthermore, we prove that under mild assumptions, this bound (i) is related to the KL divergence between the approximate and true parameter posteriors and (ii) monotonically decreases along the SVGD iterations under a suitable step-size schedule.

\MB{In summary, the contributions of this paper are fourfold: (i) the \emph{Rao--Blackwellized
Stein Gradient Descent} (RBSGD) filter, which integrates deterministic SVGD updates into a
Rao--Blackwellized framework for joint state--parameter estimation; (ii) its Fisher-informed
extension, the \emph{Rao--Blackwellized Fisher Stein Gradient Descent} (RBFSGD) filter, which
exploits curvature information for better-conditioned parameter updates; (iii) the constrained
variants cRBSGD and cRBFSGD, which enforce box, affine, and nonlinear equality constraints
through an ADMM splitting of the SVGD update; and (iv) a consistency result bounding the
state-marginal approximation error by the divergence between the true and approximate
parameter posteriors, together with its monotone tightening along the SVGD iterations.}

The remainder of this paper is organized as follows. In Section~\ref{sec:problem_statement}, we formulate the joint estimation problem for states and model parameters. Section~\ref{sec:RBPF} reviews Rao--Blackwellized Particle Filters and their application in state-space models. Building on this foundation, Section~\ref{sec:RBSGD} introduces the proposed RBSGD filter, and Section~\ref{sec:RBFSGD} presents its Fisher-informed extension, the RBFSGD filter. Section~\ref{sec:theory} establishes theoretical bounds on the posterior approximation error, and Section~\ref{sec:admm} introduces the constrained variants via an ADMM splitting. Finally, Section~\ref{sec:case_studies} demonstrates the performance of the proposed methods on three case studies: a fed-batch bioreactor and a damped pendulum under physical constraints, and a nonlinear system modelled with a neural network.

\emph{Notation:} Let $\mathbb{R}^{n}$ denote the set of real vectors with dimension $n$. The sequence of time-indexed vectors $\{u_1,\ldots,u_t\}$ is denoted as $u_{1:t}$ and $z_{t|k}$ denotes the estimate of signal $z$ at time $t$, given the observations up to and including  time $k$. The Dirac delta function centered at $\bar{\theta}$ is denoted by $\delta_{\bar{\theta}}(\theta)$. A random variable $x$ sampled from a Gaussian distribution with mean $\mu $ and covariance $\Sigma$ is denoted by $x \sim \mathcal{N}(\mu,\Sigma)$. The expected value of a random variable is denoted by $\mathbb{E}\left[x \right]$. Given two probability density functions $p$ and $q$, 
the KL divergence is denoted as $\mathrm{KL}(p||q)$, which is   defined as $\mathrm{KL}(p||q) = \int p(x) \ln \frac{p(x)}{q(x)}dx$.

\section{Problem Statement}
\label{sec:problem_statement}
We consider the following discrete-time nonlinear parameter-varying state--space model,
\begin{subequations}
\label{eq:prob_eqn}
\begin{align}
x_{t+1} &= f(x_t, u_t, \theta_{t}) + q_t, \\
y_t &= h(x_t, \theta_{t}) + r_t,
\end{align}
\end{subequations}
where $t \in \mathbb{N}$ denotes the time index,  $x_t \in \mathbb{R}^{n_x}$, $y_t \in \mathbb{R}^{n_y}$, and $u_t \in \mathbb{R}^{n_u}$ denote the system state, output, and input, respectively. The functions $f$ and $h$ describe the nonlinear state transition and measurement mappings and are parameterized by an \emph{unknown} parameter  $\theta_t \in \mathbb{R}^{n_{\theta}}$. The process and measurement noise are modeled as zero-mean Gaussian random variables, $q_t \sim \mathcal{N}(0, Q(\theta_t))$ and $r_t \sim \mathcal{N}(0, R(\theta_t))$, with corresponding positive definite covariance matrices $Q(\theta_t)$ and $R(\theta_t)$ parameterized by $\theta_t$.

Given the sequence of output  measurements $y_{1:t}$ and inputs $u_{1:t}$ up to time $t$, our objective is to estimate the joint posterior distribution over  states and time-varying model parameters $p(x_t, \theta_t |y_{1:t}, u_{1:t})$,  at time $t$. We seek a recursive filter that updates this posterior  online, from previous posterior estimate $p(x_{t-1}, \theta_{t-1}|y_{1:t-1}, u_{1:t-1})$. 
The posterior distribution can be then utilized to compute the point estimates (\emph{e.g.}, maximum a-posteriori (MAP))   of the unknown states $x_t$, and model parameters $\theta_t$, or to provide marginal state (or parameter) posterior 
to quantify the uncertainty at each time instance.

We note that in a  classical Kalman or Bayesian filtering methods, the dynamic models $f$ and $h$, as well as covariance matrices $Q,R$,  are typically assumed to be \emph{known}. The performance of the filtering algorithms is thus sensitive to the  accuracy of the model and choice of  
covariances. The estimation of the model as well as covariances is therefore critical  for accurate state estimation. In this work, we address this problem by joint estimation of the states $x_t$ and unknown parameters (dynamic model and covariances) encoded via the parameter vector $\theta_t$.   

We remark that an exact  evaluation of the posterior $p(x_t, \theta_t| u_{1:t}, y_{1:t})$ is computationally intractable, and closed-form solutions exist only  in specialized cases, \emph{e.g.}, linear time-invariant dynamics with Gaussian noises. In the following sections, we develop \emph{Rao-Blackwellized Particle 
Filtering}  algorithms that compute an \emph{approximation} of the joint posterior  $p(x_t, \theta_t| u_{1:t}, y_{1:t})$ and update it online  as new observations become available.

\section{Rao--Blackwellized filtering for joint state-parameter estimation}
\label{sec:methodologies}

The posterior  $p(x_t, \theta_t | y_{1:t}, u_{1:t})$ can be approximated using Monte Carlo sampling-based methods, which typically consists of generating a sequence of $N$ random samples $\{x^{(i)}_t, \theta^{(i)}_t \}_{i=1}^{N}$,   from a Markov chain, whose stationary distribution is the target posterior $p(x_t, \theta_t| u_{1:t}, y_{1:t})$, which can be then approximated with the following empirical distribution,
\begin{align*}
   p(x_t, \theta_t| u_{1:t}, y_{1:t}) \approx \frac{1}{N}\sum_{i=1}^{N}\delta_{\big(x^{(i)}_t,\theta^{(i)}_t\big)}\big(x_t, \theta_t \big). 
\end{align*}
A potential drawback of this naive implementation is that it requires drawing samples from a high-dimensional joint parameter space $(x, \theta)$, which is computationally restrictive for high dimensions. Instead, in this work, we employ a \emph{Rao-Blackwellization} (RB) approach in which part of the posterior is computed analytically, while only the remaining part is approximated using sampling-based algorithms. The advantage of the RB strategy is that it  reduces the size of the space over which the sampling is required. 

In particular, in the proposed RB filtering framework, after the joint posterior is factorized as,
\begin{equation}\label{eq:factor}
p(x_t, \theta_t | y_{1:t}, u_{1:t}) = p(x_t| \theta_t, y_{1:t}, u_{1:t}) \, p(\theta_t | y_{1:t}, u_{1:t}),
\end{equation}
the conditional distribution $p(x_t| \theta_t, y_{1:t}, u_{1:t})$ is approximated analytically via  Kalman filtering, while the marginal posterior over the parameters, $p(\theta_t |y_{1:t}, u_{1:t})$, is approximated using  variants of particle filters. 

Particle filters allow us to treat any type of probability distribution and nonlinearity, without being restricted to linear-Gaussian models. 
In this setting, we introduce the following two variants of the RB particle filters: \emph{Rao--Blackwellized Stein Gradient Descent} (RBSGD) and \emph{Rao--Blackwellized Fisher Stein Gradient Descent} (RBFSGD) filters.  

We first review the standard RB particle filter algorithm, and then build upon this framework to present the RBSGD and RBFSGD filters.
For brevity, we drop the dependence of the conditional distributions on inputs $u_{1:t}$ in the rest of the paper. 

\subsection*{Rao--Blackwellized Particle Filters}
\label{sec:RBPF}
In  Rao--Blackwellized Particle Filters (RBPFs), \cite{murphy2001rao} 
 a weighted set of $N$  particles $\{ w_t^{(i)}, \theta_t^{(i)}\}_{i=1}^N$, \emph{i.e.}, samples from a \emph{proposal distribution} $q(\cdot)$ and their weights,  are used to approximate the parameter posterior $p(\theta_t |y_{1:t})$ in \eqref{eq:factor}  by the empirical
point-mass distribution as follows, 
\begin{align}\label{eq:parameter_post}
    p(\theta_t |y_{1:t}) \approx \sum_{i=1}^{N} w^{(i)}_t\delta_{\theta^{(i)}_t}(\theta_t), \quad \sum_{i=1}^{N}w^{(i)}_t =1,  
 \end{align}
and for each particle $\theta^{(i)}$ (representing the $i$-th model), an \emph{Extended  Kalman Filter} (EKF) is employed to compute the corresponding conditional state distribution $p(x_t^{(i)} | \theta_t^{(i)}, y_{1:t})$. 
Note that marginalizing out the state variable allows analytic computation of state posterior,   $p(x_t^{(i)} | \theta_t^{(i)}, y_{1:t})$ which is the Gaussian posterior distribution derived from EKF given as,
\begin{align}\label{eq:Kalman_post}
   p(x_t^{(i)} | \theta_t^{(i)}, y_{1:t}) = \mathcal{N}(x^{(i)}_t; \hat{x}^{(i)}_{t|t}, P^{(i)}_{t|t}),  
\end{align}
where  $\hat{x}_{t|t}^{(i)}, P_{t|t}^{(i)}$  denote the  mean and covariance of the state associated with the $i$-th particle respectively. 

RBPFs proceed in three stages: prediction, update, and resampling, as summarized in Algorithm~\ref{alg:rbpf_simple}. At the initial Step~\ref{alg:s0}, the particles and their weights as well as Kalman filter's mean and covariances are initialized according to prior distributions and initial values.  

In the prediction Step~\ref{algo:s1a}-\ref{algo:s1b}, both state and parameter distributions are propagated forward in time. The parameters evolve according to the following stationary stochastic transition model with additive Gaussian noise,
\begin{equation}
\label{eq:theta_transition}
\theta_t^{(i)} = \theta_{t-1}^{(i)} + \eta_t^{(i)}, \quad \eta_t^{(i)} \sim \mathcal{N}(0, \Sigma_\theta),
\end{equation}
where $\Sigma_\theta$ controls the degree of parameter drift, allowing adaptation to time-varying dynamics. 
In common practice,  the parameter proposal distribution $q(\theta_t|\theta^{(i)}_{t-1})$ is taken as the transition model $p(\theta_t|\theta^{(i)}_{t-1})$. Thus, the proposal  $q(\theta_t|\theta^{(i)}_{t-1})$ is  Gaussian,  with mean $\theta^{(i)}_{t-1}$ and covariance $\Sigma_\theta$. For each particle $\theta^{(i)}_t$, an EKF is run for state estimation. The state prediction in the EKF follows the system dynamics by linearizing the model \eqref{eq:prob_eqn} as follows, 
\begin{subequations}
\label{eq:ekf_prediction}
\begin{align}
\hat{x}_{t|t-1}^{(i)} &= f\big(\hat{x}_{t-1|t-1}^{(i)}, u_{t-1}, \theta_{t-1}^{(i)}\big), \\
P_{t|t-1}^{(i)} &= F_{t-1}^{(i)} P_{t-1|t-1}^{(i)} (F_{t-1}^{(i)})^\top + Q(\theta_{t-1}^{(i)}),
\end{align}
\end{subequations}
where $\hat{x}_{t|t-1}^{(i)}$ and $\hat{x}_{t-1|t-1}^{(i)}$ denote the predicted and previous state means,  $P_{t|t-1}^{(i)}$, $P_{t-1|t-1}^{(i)}$ are the corresponding state covariance matrices, and  $F_{t-1}^{(i)} \!=\! \frac{\partial{f}}{\partial{x}}$ is the Jacobian of the dynamics $f$, evaluated at $\hat{x}_{t-1|t-1}^{(i)}$ and $\theta_{t-1}^{(i)}$.

Upon receiving the observation $y_t$, the Kalman gain for each particle is computed as
\begin{equation}
\label{eq:kalman_gain}
K_t^{(i)} = P_{t|t-1}^{(i)} (H_t^{(i)})^\top \big( H_t^{(i)} P_{t|t-1}^{(i)} (H_t^{(i)})^\top + R(\theta_{t-1}^{(i)}) \big)^{-1},
\end{equation}
where  Jacobian $H_t^{(i)} \!=\! \frac{\partial{h}}{\partial{x}}$ of the measurement model $h$ is evaluated at $\hat{x}_{t|t-1}^{(i)}$ and $\theta_{t-1}^{(i)}$, and the innovation covariance is defined as
\begin{equation}
S_t^{(i)} = H_t^{(i)} P_{t|t-1}^{(i)} (H_t^{(i)})^\top + R(\theta_{t-1}^{(i)}).
\end{equation}

Next, in the update Step~\ref{algo:s2a}, the mean  and covariance update of the state $x^{(i)}_t$, which form the Gaussian posterior \eqref{eq:Kalman_post}
are computed  as,
\begin{subequations}
\label{eq:ekf_update}
\begin{align}
\hat{x}_{t|t}^{(i)} &= \hat{x}_{t|t-1}^{(i)} + K_t^{(i)} \big( y_t - h(\hat{x}_{t|t-1}^{(i)}, \theta_{t-1}^{(i)}) \big), \\
P_{t|t}^{(i)} &= \big( I - K_t^{(i)} H_t^{(i)} \big) P_{t|t-1}^{(i)}.
\end{align}
\end{subequations}

In the particle weight update Step~\ref{algo:s2b}, the importance weights are updated according to the likelihood of the new measurement as~\footnote{In the derivation of  RBPF, we typically work on the full trajectory posterior $p(x_{0:t}, \theta_{0:t}| y_{1:t})$ and then marginalize to get $p(x_t,\theta_t|y_{1:t}).$  },
\begin{subequations}
\begin{align}
\label{eq:weight_update}
&w_t^{(i)} \propto w_{t-1}^{(i)} \, p\big(y_t \mid \theta_{t-1}^{(i)}, y_{1:t-1}\big) \frac{p(\theta^{(i)}_t|\theta^{(i)}_{t-1})}{q(\theta^{(i)}_t|\theta^{(i)}_{t-1})}, \\
\label{eq:liklihood}
&p\big(y_t \mid  \theta_{t-1}^{(i)}, y_{1:t-1}\big) = \mathcal{N}\big(y_t; h(\hat{x}_{t|t-1}^{(i)}, \theta_{t-1}^{(i)}), S_t^{(i)}\big).
\end{align}
\end{subequations}


To prevent particle degeneracy, \emph{resampling} Step~\ref{algo:s3} is performed when the effective sample size (ESS) falls below a given threshold $n_{\text{thr}}$. Particles with higher weights are replicated by sampling from a new probability distribution $\theta \sim \sum_{i=1}^{N}w^{(i)}_t\delta_{\theta_{t}^{(i)}}(\theta)$, while those with negligible weights are discarded, and all resampled particles are assigned uniform weights, $w_t^{(i)} = 1/N$, preserving a consistent particle population.  

In summary, the RBPF presented in Algorithm~\ref{alg:rbpf_simple} produces at each time step $t$ a set of weighted samples $\{w^{(i)}_t,\, \theta^{(i)}_t,\, \hat{x}^{(i)}_{t|t},\, P^{(i)}_{t|t}\}_{i=1}^{N}$.
Then, from \eqref{eq:parameter_post} and \eqref{eq:Kalman_post}, the joint posterior in \eqref{eq:factor} is approximated with RBPF as,
\begin{align}\label{eq:approx_posterior}
    p(x_t, \theta_t|y_{1:t}) \approx \sum_{i=1}^{N} w^{(i)}_t \delta_{\theta^{(i)}_t}(\theta_t)\mathcal{N}(x_t;\hat{x}^{(i)}_{t|t}, P^{(i)}_{t|t}).
\end{align}

\begin{algorithm}[t!]
\caption{Rao--Blackwellized Particle Filter}
\begin{algorithmic}[1]
\Require Number of particles $N$, priors for initial parameter and state: $p(\theta_0), p(x_0)$; initial state covariance $\bar{P}_0$; proposal distribution $q(\theta_t|\theta_{t-1})$, resampling threshold $n_{\text{thr}}$.
\State \label{algo:s0} \textbf{Initialization:}\label{alg:s0}
\Statex \textbf{for} {$i = 1$ to $N$}
    \Statex \quad Sample $\theta_0^{(i)} \sim p(\theta_0)$
    \Statex \quad Set $\hat{x}_0^{(i)} \sim p(x_{0}),\, P_0^{(i)} = \bar{P}_{0}$   \Comment{Initial KF state}
    \Statex \quad $w_0^{(i)} \leftarrow 1/N$
\Statex \textbf{end for} 
\Statex \textbf{for} {$t = 1$ to $T$}
\Statex \quad \textbf{for} {$i = 1$ to $N$}
\State \label{algo:s1a} \quad \quad\textbf{Particle prediction:}
         $\theta_t^{(i)} \sim q(\theta_t \mid \theta_{t-1}^{(i)})$ \eqref{eq:theta_transition} 
 \State \label{algo:s1b} \quad \quad \textbf{Kalman prediction:} Compute   $\hat{x}_{t|t-1}^{(i)}$, $P_{t|t-1}^{(i)}$~\eqref{eq:ekf_prediction}
 \State \label{algo:s2b} \quad \quad  \textbf{Weight update:}
       $\tilde{w}_t^{(i)}$ as in \eqref{eq:weight_update}
\State \label{algo:s2a} \quad  \quad \textbf{Kalman  update:}
        Compute 
       $\hat{x}_{t|t}^{(i)}$, $P_{t|t}^{(i)}$: eq.~\eqref{eq:ekf_update}
\Statex \quad \textbf{end for}
\Statex \quad \quad \ {Normalize} $w_t^{(i)} = 
           \tilde{w}_t^{(i)} \big/ \sum_{j=1}^N \tilde{w}_t^{(j)}, \ i=1,\ldots, N.$  
        \State \quad \quad \label{algo:s3} \textbf{Resample} particles according to the weight $w_t^{(i)}$ if $\mathrm{ESS} = 1 / \sum_{i=1}^N (w_t^{(i)})^2 < n_{\text{thr}}$
        \Statex \quad \quad \quad \textbf{for} {$i = 1$ to $N$}
            \Statex \quad \quad \quad \quad \textbf{resample} $\theta_t^{(i)} \sim \sum_{j=1}^{N}w^{(j)}_t\delta_{\theta_t^{(j)}}(\theta)$
            \Statex \quad \quad \quad \quad \textbf{set} $w_t^{(i)} \leftarrow 1/N$
\Statex \quad \quad \quad \textbf{end for} 
\Statex \textbf{end for} 
\Statex \Return $\{\{\hat{x}_{t|t}^{(i)}, P_{t|t}^{(i)}, \theta_t^{(i)}\}_{i=1}^N\}_{t=1}^T$
\end{algorithmic}
\label{alg:rbpf_simple}
\end{algorithm}

\section{Rao--Blackwellized Filtering with Stein Gradient Descent}

The RBPF framework discussed in the previous section combines analytical filtering for state estimation with a particle filters approximation for parameter inference, achieving a balance between computational tractability and flexibility in representing non-Gaussian parameter distributions. 

However, the performance of the algorithm critically depends on choosing an appropriate  proposal distribution $q(\theta_t|\theta_{t-1})$, which  can be chosen based on the   stochastic transition model \eqref{eq:theta_transition} for the parameters. This, in turn, 
requires
careful tuning of the  noise covariance $\Sigma_\theta$, to reflect the  temporal variability in the parameters. A high-variance proposal spreads particles into low-probability regions, causing weight degeneracy, while a low-variance proposal restricts exploration and yields an  overly concentrated and inaccurate posterior approximation.
Determining the optimal proposal parameters for a specific target posterior distribution is therefore nontrivial in particle filtering. In addition, accurate approximation of the parameter posterior \eqref{eq:parameter_post} may require a large number of particles, which can lead to significant computational overhead. 

To address these challenges, we adopt a \emph{deterministic}  \emph{Stein Variational Gradient Descent} (SVGD) \cite{liu2016stein} framework, where particles evolve via a gradient-based update that minimizes the Kullback–Leibler (KL) divergence between the true posterior and its particle approximation. This automatically balances particle exploration and concentration in high-posterior-density regions, eliminating the need to tune the variance parameter $\Sigma_{\theta}$ to trade off these effects.  
In the next sections, we introduce two variants of RBPFs where the SVGD method~\cite{liu2016stein} is employed to govern the particle evolution and approximation of the parameter posterior $p(\theta_t|y_{1:t})$ within a Rao-Blackwellization framework, offering an efficient alternative to the classical RBPFs.  

\subsection{Rao--Blackwellized Stein Gradient Descent Filter}
\label{sec:RBSGD}

 In this section, we discuss how to  integrate SVGD into the Rao-Blackwellization framework with a Kalman filter for joint  state-parameter posterior computation. 

The SVGD~\cite{liu2016stein} is a deterministic, particle-based variational inference method for approximating posterior distributions. In SVGD, instead of a  stochastic transition model as in $\eqref{eq:theta_transition}$, 
the particles $\theta^{(i)}_t$  
evolve via a specific \emph{transport map}, \emph{i.e.}, a small perturbation of the identity map: $ \theta + \epsilon \phi(\theta)$, where $\epsilon$ is a small constant and
$\phi$ is a smooth function.  
This transformation iteratively transports the empirical proposal distribution $\theta \sim q(\theta) = \frac{1}{N}\sum_{i=1}^{N}\delta_{\theta^{(i)}}(\theta)$ of the particles toward the target posterior $p(\theta|y_{1:t})$.
Specifically, SVGD considers the following  transformation of the particles:
\begin{equation}
\label{eq:svgd_transform}
T_{\epsilon \phi}:\theta^{(i)} \mapsto \theta^{(i)} + \epsilon \, \phi_{q,p}^\ast(\theta^{(i)}),
\end{equation}
where $\epsilon >0$ is a step size representing perturbation magnitude and ${\phi}_{q,p}^{\ast}(\cdot) \in \mathcal{H}$ is   the optimal perturbation direction that  lies in a vector-valued Reproducing Kernel Hilbert Space $\mathcal{H}$ (RKHS\footnote{The RKHS of a positive definite kernel $k(x,x')$ is the closure of the linear span $\{f: f(x) = \sum_{i=1}^{m}a_i k(x,x_i), a_i \in \mathbb{R}\}$, equipped with inner products $\langle f,g \rangle_{\mathcal{H}} = \sum_{ij}a_i b_j k(x_i, x_j)$ for $g(x) = \sum_ib_ik(x,x_i)$.}) with associated positive definite kernel  $k(\theta, \theta')$. 
We next   explain the intuition behind ${\phi}_{q,p}^\ast(\cdot)$ and show that it 
maximally decreases the KL divergence between the proposal $q$ and target posterior $p$.  
The optimal perturbation direction ${\phi}_{q,p}^{\ast}(\theta)$ in \eqref{eq:svgd_transform} (which depends on both proposal and target distributions $q,p$  as emphasized in the subscript), is given by  the following expectation~\cite{liu2016stein},
\begin{equation}
\label{eq:svgd_phi_pi}
\begin{split}
\phi_{q,p}^\ast(\theta) = 
\mathbb{E}_{\theta' \sim q(\theta)} \Big[
&k(\theta', \theta) \, \nabla_{\theta'} \log p(\theta'|y_{1:t}) \\
&+ \nabla_{\theta'} k(\theta', \theta)
\Big],
\end{split}
\end{equation}
where $k(\theta,\theta')$ is a user-defined positive-definite kernel, \emph{e.g.}, the radial basis function $k(\theta, \theta') = \exp\!\big(\frac{-1}{h}\|\theta \!-\! \theta'\|^2\big)$, and $\nabla_{\theta'}\log p(\theta'|y_{1:t})$ is the gradient of the log-posterior evaluated at  $\theta'$. In \eqref{eq:svgd_phi_pi}, the first term drives the particles toward regions of high posterior probability, while the second term introduces repulsion to maintain particle spread, preventing all particles from collapsing into local modes of the posterior. This automatically balances particle exploration with concentration in high-density areas.

It can be proved that  $\phi_{q,p}^\ast$ in \eqref{eq:svgd_phi_pi} is the optimal perturbation direction that gives the steepest descent on the 
KL divergence between proposal and target posterior. The result is formally stated in the following lemma from~\cite{liu2016stein}. 

\begin{lemma}\label{thm:svgd} Consider the transformation $T_{\epsilon \phi}(\theta) = \theta + \epsilon \phi(\theta)$ in \eqref{eq:svgd_transform}, and transformed particles $z = T_{\epsilon \phi}(\theta)$ with $\theta \sim q(\theta)$. Let $q_{[T_{\epsilon \phi}]}(z)$ denote the probability density of the particles after this transformation\footnote{The  density for $z = T(\theta)$ is  $q_{[T]}(z)= q(T^{-1}(z))\cdot | \mathrm{det}((\nabla_zT^{-1}(z))|$.}. Then at the population limit ($N \to \infty$), the  perturbation direction $\phi_{q,p}^\ast$ in \eqref{eq:svgd_phi_pi} is,
\begin{equation}
\label{eq:phi_opt_pi}
\phi_{q,p}^\ast = \arg\max_{\phi \in \mathcal{H}, \|\phi \|_{\mathcal{H}}\leq 1} \left\{ -\frac{d}{d\epsilon} \mathrm{KL}(q_{[T_{\epsilon\phi}]} \,\|\, p) \Big|_{\epsilon = 0} \right\},
\end{equation}
i.e., $\phi_{q,p}^{\ast}$ is the  direction of the steepest descent that maximizes the negative  gradient of the KL term.      
\end{lemma}
We refer the reader to \cite[Theorem 3.1, Lemma 3.2]{liu2016stein} for details and proofs.

Lemma~\ref{thm:svgd} states that the optimal perturbation direction $\phi_{q,p}^{\ast}$ maximally decreases the KL divergence between the  transformed proposal $q$  after the infinitesimal transformation of the particles $\theta \mapsto \theta + \epsilon \phi(\theta)$ and the target posterior $p$.  This result allows for an \emph{iterative} procedure that transforms initial proposal  $q_0(\theta)$ towards the target $p(\theta)$: starting from applying the transform $T_0(\theta) = \theta + \epsilon \phi^{\ast}_{q_0,p}$ on $q_0$ decreases the KL divergence. This would give a new distribution $q_{1}$, on which again applying the transformation $T_1(\theta) = \theta + \epsilon \phi^{\ast}_{q_1,p}$ would further decrease the KL divergence. Repeating this process results in a sequence of distributions $\{q_m\}_{m=1}^{M}$ between $q_0$ and $p$ via $T_{m}(\theta) = \theta + \epsilon \phi^{\ast}_{q_m,p}(\theta)$ for $m=1,\ldots, M$. 
The following proposition establishes that, under standard smoothness assumptions, each SVGD step decreases the KL
divergence to the target posterior for sufficiently small step sizes.

\begin{proposition}
\label{thm:svgd_one_step_formal}
Assume that target distribution $\pi_t(\theta) =p(\theta_t|y_{1:t})$ at time $t$ is twice continuously differentiable with bounded Hessian
and that the SVGD kernel $k(\theta_t,\theta_t')$ is bounded\MB{, smooth, and $C_0$-universal}. Consider a sequence of proposal distributions   $\{q_{t,m}\}_{m\geq 0}$  where $q_{t,m}$ denotes the proposal distribution at time $t$ for the $m$-th iteration.   
Then at the population limit ($N \to \infty$), when $q_{t,m-1}\ne \pi_t$, there exists
\MB{$\epsilon_{0,m-1}>0$}, such that for all step sizes $\epsilon: 0<\epsilon<\MB{\epsilon_{0,m-1}}$,
\begin{align}
\label{eq:KL_statement}
\mathrm{KL}\!\left(q_{t,m}\,\|\,\pi_t\right)
< 
\mathrm{KL}\!\left(q_{t,m-1}\,\|\,\pi_t\right). 
\end{align}
\end{proposition}


We refer to Appendix~\ref{proof:svgd_one_step_formal} for the proof of Proposition~\ref{thm:svgd_one_step_formal}.


The result in Proposition~\ref{thm:svgd_one_step_formal} allows to formulate the following iterative update of the particles at time $t$, 
\begin{equation}
\label{eq:svgd_update}
\theta^{(i)}_{t,m} = \theta^{(i)}_{t,m-1} + \epsilon \, \hat{\phi}^\ast(\theta^{(i)}_{t,m-1}),
\end{equation}
where $\theta^{(i)}_{t,m}$ denotes the $i$-th particle at iteration $m$, at time $t$ and $\hat{\phi}^{\ast}(\cdot)$ is an \emph{empirical approximation} of  the perturbation direction in \eqref{eq:svgd_phi_pi}, 
computed based on averages taken over $N$ particles as follows,
\begin{align}
\hat{\phi}^\ast(\theta) = 
\frac{1}{N} \sum_{j=1}^N &
\big[k(\theta^{(j)}_{t,m},\theta)\,\nabla_{\!\theta^{(j)}_{t,m}} \log p(\theta^{(j)}_{t,m}|y_{1:t}) \notag\\
&\quad   + \nabla_{\!\theta^{(j)}_{t,m}} k(\theta^{(j)}_{t,m},\theta)\big].
\label{eq:svgd_empirical}
\end{align}
Note that the target posterior is unknown due to its intractable normalization constant, however, as the gradient of the log-posterior is unaffected by additive constants, $\nabla_{\!\theta^{(j)}_{t,m}} \log p(\theta^{(j)}_{t,m}|y_{1:t})$ in \eqref{eq:svgd_empirical} can be computed 
directly from the unnormalized posterior as follows\footnote{dropping the time and iteration index $t,m$ for brevity},
\begin{align*}
\nabla_{\theta} \log p(\theta|y_{1:t}) 
\!=\! \nabla_{\theta} \log p(y_t| \theta, y_{1:t\!-\!1}) 
\!+\! \nabla_{\theta} \log p(\theta|y_{1:t\!-\!1}),
\end{align*}
where the first term $p(y_t|\theta,y_{1:t-1})$ is the marginal likelihood computed from the
Kalman filter iterations as in \eqref{eq:liklihood}, and the second term is the score of the
prior on $\theta$ at time $t$, i.e.\ the previous posterior $p(\theta\mid y_{1:t-1})$. \MB{At
$t=1$ this is the score of the initial prior $p(\theta)$; for $t>1$ we approximate the
previous posterior by a Gaussian $\mathcal{N}(\theta;\,\hat\theta_{t-1},\Sigma_{t-1})$ fitted
to the particles $\{\theta^{(i)}_{t-1}\}_{i=1}^{N}$, giving the closed-form recursive score
\begin{equation}
\label{eq:prior_score}
\nabla_{\theta}\log p(\theta\mid y_{1:t-1}) \approx -\Sigma_{t-1}^{-1}(\theta - \hat\theta_{t-1}),
\end{equation}
where $\hat\theta_{t-1}$ and $\Sigma_{t-1}$ are the mean and covariance of the previous
particle set. The covariance $\Sigma_{t-1}$ is \emph{computed} from the particles rather than
tuned, in contrast to the proposal covariance $\Sigma_\theta$ of the RBPF, and this
assumed-density recursion~\cite{Kotecha2003} carries the accumulated posterior information
forward at $O(1)$ cost per step.}
Thus, at the next SVGD iteration $m$, the particles   are updated via the empirical version of the deterministic transformation $T_m(\theta^{(i)}_{t,m}) := \theta_{t,m-1}^{(i)} +\epsilon \hat{\phi}^{\ast}(\theta^{(i)}_{t,m-1})$ as in \eqref{eq:svgd_update} with perturbation direction as in  \eqref{eq:svgd_empirical}.



We now integrate SVGD into our Rao-Blackwellization framework to approximate the joint state-parameter posterior $p(x_t,\theta_t|y_{1:t})$. The iterative algorithm is referred to as \emph{Rao-Blackwellized Stein Gradient Descent} (RBSGD). The RBSGD procedure at time step $t$ is summarized in Algorithm~\ref{alg:RBSGD}. 
In the RBSGD Algorithm~\ref{alg:RBSGD}, at each time step, the conditional state mean and covariance are first  updated via the EKF, followed by a deterministic SVGD update of the parameter particles for $M$ iterations.
 In particular,  in the Kalman prediction and update steps   state mean and covariance $\hat{x}_{t|t}, P_{t|t}$ are computed.  Next, starting from the particle initialization $\theta^{(i)}_{t-1}$ obtained at $t-1$,  the  particles $\{\theta_t^{(i)}\}_{i=1}^{N}$ are iteratively updated  according to the SVGD transformation \eqref{eq:svgd_update} for $M$ iterations, where the perturbation direction $\hat{\phi}^{\ast}$   depends on the gradient of the log-likelihood of the Kalman filter as in~\eqref{eq:liklihood} and a user-chosen kernel function $k(\theta, \theta')$. 
 For each particle, the conditional state distribution $p(x_t^{(i)}| \theta_t^{(i)}, y_{1:t})$ is computed analytically as in~\eqref{eq:Kalman_post} using an EKF. Consequently, the RBSGD maintains the Rao–Blackwellized decomposition, but replaces the stochastic propagation of parameters with a deterministic SVGD-based update, thereby eliminating the need for manual tuning of $\Sigma_\theta$ in \eqref{eq:theta_transition}. The resulting algorithm combines the analytical efficiency of Kalman filtering with the deterministic convergence properties of SVGD, offering improved sample efficiency and stability in tracking time-varying parameters.

\begin{algorithm}[h!]
\caption{Rao--Blackwellized Stein Gradient Descent Filter}
\begin{algorithmic}
\Require Measurement $y_t$, State mean and covariance $\{\hat{x}_{t-1|t-1}^{(i)}, P_{t-1|t-1}^{(i)}\}_{i=1}^N$, previous particles $\{\theta_{t-1}^{(i)}\}_{i=1}^N$, RBF kernel $k(\theta,\theta')$, SVGD step size $\epsilon >0$
\For{$i = 1$ to $N$}
    \State \textbf{Kalman Prediction:} $\hat{x}_{t|t-1}^{(i)}$, $P_{t|t-1}^{(i)}$ using \eqref{eq:ekf_prediction}
    \State \textbf{Kalman Update:} $\hat{x}_{t|t}^{(i)}$, $P_{t|t}^{(i)}$ using \eqref{eq:ekf_update}
\EndFor
\State \textbf{Particle Update:}
  \State  $\theta^{(i)}_{t,0} \leftarrow  \theta^{(i)}_{t-1},\quad i=1,\ldots, N$
\For{$i = 1$ to $N$}
    \For{$m = 1$ to $M$}
        \State $\theta_{t,m}^{(i)} \leftarrow \theta_{t,m-1}^{(i)} + \epsilon \, \hat{\phi}^\ast(\theta_{t,m-1}^{(i)})$ using \eqref{eq:svgd_update}--\eqref{eq:svgd_empirical} \footnotemark
    \EndFor
    \State $\theta_{t}^{(i)} \leftarrow  \theta_{t,M}^{(i)}$
\EndFor
\State \Return $\{\hat{x}_{t|t}^{(i)}, P_{t|t}^{(i)}, \theta_t^{(i)}\}_{i=1}^N$
\end{algorithmic}
\label{alg:RBSGD}
\end{algorithm}
\footnotetext{Parameter update can be optionally done with an adaptive step size $\epsilon$ provided by the Adam algorithm.}

\subsection{Rao--Blackwellized Fisher Stein Gradient Descent Filter}
\label{sec:RBFSGD}

While the RBSGD filter mitigates particle degeneracy and 
enhances stability by eliminating the dependence on the proposals
 with deterministic SVGD updates, its \emph{first-order} optimization dynamics \eqref{eq:svgd_update} may be slow or poorly conditioned when the posterior exhibits strong curvature or parameter correlations. To address this, we introduce the \emph{Rao--Blackwellized Fisher Stein Gradient Descent (RBFSGD)} filter, which incorporates \emph{second-order} information via the \emph{Fisher Information Matrix} (FIM). Similar to natural-gradient and Fisher-based optimization methods \cite{amari1998natural, hwang2024fadam}, RBFSGD uses the FIM as a positive semi-definite curvature approximation of the log-likelihood, preconditioning gradient directions to achieve better-conditioned, directionally informed updates without the computational cost of full Hessian evaluation \cite{schraudolph2002fast, martens2010deep, vinyals2012krylov}.

In FSGD, the SVGD particle update in~\eqref{eq:svgd_update} is modified by incorporating
second-order information as follows \cite{hwang2024fadam},
\MB{
\begin{align}\label{eq:svgd_update_FIM}
    \theta^{(i)}_{t,m} = \theta^{(i)}_{t,m-1} + \epsilon \, \mathcal{F}^{-1}(\theta^{(i)}_{t,m-1}) \, \hat{\phi}^{\ast}(\theta^{(i)}_{t,m-1})
\end{align}
}
where $\mathcal{F}(\theta) \in \mathbb{R}^{n_{\theta} \times n_{\theta}}$ is the symmetric
positive semi-definite FIM with respect to the parameters $\theta$, defined from the marginal
likelihood as
\begin{align}
\mathcal{F}(\theta) 
&= \mathbb{E}_{y_t \sim p(y_t| \theta, y_{1:t-1})} \Big[
\nabla_\theta \log p(y_t| \theta,y_{1:t-1}) \notag\\
&\quad\quad\quad \times
\big(\nabla_\theta \log p(y_t| \theta, y_{1:t-1})\big)^\top
\Big].
\label{eq:fim_def_meas}
\end{align}
\MB{Since $\mathcal{F}$ is only positive semi-definite and hence possibly singular, the
preconditioner in \eqref{eq:svgd_update_FIM} is not formed by inverting $\mathcal{F}$
directly; it is realized through a damped second-moment estimate.}
\MB{For the EKF, the one-step marginal likelihood $p(y_t| \theta, y_{1:t-1})$ is Gaussian
\eqref{eq:liklihood}, so its score $\nabla_\theta \log p(y_t| \theta, y_{1:t-1})$ is available
in closed form. The expectation over $y_t$ in \eqref{eq:fim_def_meas} is, however,
intractable, so we adopt the empirical Fisher approximation: the expectation is replaced by
the outer product of the score at the realized measurement, averaged over the particle cloud,
\begin{align}
\widehat{\mathcal{F}}_{\mathrm{lik}}
&= \frac{1}{N} \sum_{i=1}^N 
\nabla_{\theta^{(i)}} \log p(y_t| \theta^{(i)}, y_{1:t-1}) \notag\\
&\quad\quad\quad \times
\big(\nabla_{\theta^{(i)}} \log p(y_t| \theta^{(i)}, y_{1:t-1})\big)^\top.
\label{eq:fim_empirical_lik}
\end{align}
Each score is evaluated analytically from the Kalman innovation and its sensitivities with
respect to $\theta$, so $\widehat{\mathcal{F}}_{\mathrm{lik}}$ captures the local curvature of
the Gaussian log-likelihood \eqref{eq:liklihood} without sampling.} The empirical FIM
\eqref{eq:fim_empirical_lik} is then used to construct the \emph{Fisher-informed} RBF kernel
$k_{\mathcal{F}}(\theta, \theta')$,
\begin{equation}
\label{eq:fisher_rbf}
k_\mathcal{F}(\theta^{(i)}, \theta^{(j)}) = \exp\!\Big(\! -\frac{1}{h} (\theta^{(i)} - \theta^{(j)})^\top \widehat{\mathcal{F}}_{\mathrm{lik}} (\theta^{(i)} - \theta^{(j)}) \! \Big),
\end{equation}
which, in turn, enters the computation of the perturbation direction $\hat{\phi}^{\ast}(\theta)$
in the FSGD update~\eqref{eq:svgd_update_FIM}. Compared with the RBF kernel $k(\theta, \theta')$
used in SVGD, the Fisher-informed kernel $k_\mathcal{F}(\theta^{(i)}, \theta^{(j)})$ in
\eqref{eq:fisher_rbf} adapts the similarity metric between particles according to the local
curvature, accounting for anisotropic parameter sensitivities.

\MB{Analogously, a second-moment matrix of the SVGD update directions is defined as
\begin{equation}
\label{eq:fim_empirical_svgd}
\widehat{\mathcal{F}}_{\mathrm{svgd}} 
= \frac{1}{N} \sum_{i=1}^N \hat{\phi}^{\ast}(\theta^{(i)}) \hat{\phi}^{\ast}(\theta^{(i)})^\top,
\end{equation}
where $\hat{\phi}^{\ast}(\theta^{(i)})$ is the update direction of \eqref{eq:svgd_empirical}
with the RBF kernel replaced by the Fisher-informed kernel \eqref{eq:fisher_rbf}. Unlike
$\widehat{\mathcal{F}}_{\mathrm{lik}}$, which approximates the likelihood curvature,
$\widehat{\mathcal{F}}_{\mathrm{svgd}}$ is the second moment of the search directions and
serves as an adaptive, Adam-style preconditioner in the update~\eqref{eq:svgd_update_FIM}.}

We now present the FSGD algorithm for parameter updates incorporating the Fisher information
matrices \eqref{eq:fim_empirical_lik}--\eqref{eq:fim_empirical_svgd}. The parameter particles
are updated using a Fisher-preconditioned variant of the Adam optimizer, known as
Fisher--Adam \cite{hwang2024fadam}. The update for particle $\theta^{(i)}$ at the $m$-th
iteration of time step $t$ is
\MB{
\begin{equation}
\label{eq:fisher_adam_step}
\theta_{t,m}^{(i)} = \theta_{t,m-1}^{(i)} + \epsilon \, (l_{t,m})^{-\top} (l_{t,m})^{-1} \hat{g}_{t,m}^{(i)},
\end{equation}
}
where $\epsilon$ is the step size and the auxiliary variables evolve as
\MB{
\begin{subequations}
\label{eq:fisher_adam_aux}
\begin{align}
g_{t,m}^{(i)} &= \beta_1 g_{t,m-1}^{(i)} + (1 - \beta_1)\, \hat{\phi}^{\ast}(\theta^{(i)}_{t,m-1}), \\
v_{t,m} &= \beta_2 v_{t,m-1} + (1 - \beta_2)\, \widehat{\mathcal{F}}_{\mathrm{svgd}}, \\
\hat{g}_{t,m}^{(i)} &= \frac{g_{t,m}^{(i)}}{1-{\beta_1}^m}, \\
\hat{v}_{t,m} &= \frac{v_{t,m}}{1-{\beta_2}^m}, \\
l_{t,m} l_{t,m}^{\top} &= \hat{v}_{t,m} + \lambda I.
\end{align}
\end{subequations}
}
\MB{Here, $g_{t,m}^{(i)}$ is the per-particle first moment of the update direction
$\hat{\phi}^{\ast}(\theta^{(i)})$, and $v_{t,m}$ is the shared second moment driven by
$\widehat{\mathcal{F}}_{\mathrm{svgd}}$; $\hat{g}_{t,m}^{(i)}$ and $\hat{v}_{t,m}$ are their
bias-corrected counterparts, and $l_{t,m}$ is the Cholesky factor of $\hat{v}_{t,m}+\lambda I$.
The direction $\hat{\phi}^{\ast}(\theta^{(i)})$ is the Fisher--Stein gradient descent (FSGD)
update direction, i.e.\ the SVGD direction \eqref{eq:svgd_empirical} with the Fisher-informed
kernel \eqref{eq:fisher_rbf}, and $\widehat{\mathcal{F}}_{\mathrm{svgd}}$ is the second-moment
matrix \eqref{eq:fim_empirical_svgd}. The decay rates $\beta_1,\beta_2$ are typically set to
$0.9$ and $0.999$. The damping constant $\lambda>0$ keeps $\hat{v}_{t,m}+\lambda I$ positive
definite, which is needed since $\widehat{\mathcal{F}}_{\mathrm{svgd}}$ has rank at most $N$
and is therefore singular whenever $N<n_\theta$. The moments are reset,
$g^{(i)}_{t,0}=0$ and $v_{t,0}=0$, at the start of each time step, so the bias correction with
exponent $m$ applies to the $M$-step inner run. This formulation lets the parameter updates
incorporate local curvature efficiently, improving conditioning and stability over the
first-order SVGD update \eqref{eq:svgd_update}.}

As in RBSGD, we integrate the FSGD particle update \eqref{eq:fisher_adam_step}  in the Rao-Blackwellized framework. The resulting  RBFSGD procedure at time $t$ is summarized in Algorithm~\ref{alg:RBFSGD}. At each time step $t$, the conditional state mean $\hat{x}_{t|t}$ and covariance $P_{t|t}$ are first updated via the Kalman filter, followed by calculation of the FIMs \eqref{eq:fim_empirical_lik}--\eqref{eq:fim_empirical_svgd}. Finally, the  parameters are updated by taking  $M$ Fisher--Adam steps as in \eqref{eq:fisher_adam_step}.

\begin{algorithm}[!h]
\caption{Rao--Blackwellized Fisher Stein Gradient Descent Filter}
\begin{algorithmic}
\Require Measurement $y_t$, EKF mean state and covariance$\{\hat{x}_{t-1|t-1}^{(i)}, P_{t-1|t-1}^{(i)}\}_{i=1}^N$, particles $ \{\theta_{t-1}^{(i)}\}_{i=1}^{N}$, Fisher informed kernel $k_{\mathcal{F}}(\cdot,\cdot)$, step size $\epsilon$
\For{$i = 1$ to $N$}
    \State \textbf{Kalman Predict} $\hat{x}_{t|t-1}^{(i)}$, $P_{t|t-1}^{(i)}$ via \eqref{eq:ekf_prediction}
    \State \textbf{Kalman Update} $\hat{x}_{t|t}^{(i)}$, $P_{t|t}^{(i)}$ via \eqref{eq:ekf_update}
\EndFor
\State \textbf{Compute} $\widehat{\mathcal{F}}_{\mathrm{lik}}$ and $\widehat{\mathcal{F}}_{\mathrm{svgd}}$ using \eqref{eq:fim_empirical_lik} and \eqref{eq:fim_empirical_svgd}
\State \textbf{Particle Update:} 
  \State  $\theta^{(i)}_{t,0} \leftarrow  \theta^{(i)}_{t-1},\; g^{(i)}_{t,0}\leftarrow 0,\quad i=1,\ldots, N;\quad v_{t,0}\leftarrow 0$
\For{$i = 1$ to $N$}
    \For{$m = 1$ to $M$}
        \State Update $\theta_{t, m}^{(i)}$ using Fisher--Adam step \eqref{eq:fisher_adam_step}
    \EndFor
    \State $\theta_{t}^{(i)} \leftarrow \theta_{t,M}^{(i)}$
\EndFor
\State \Return $\{\hat{x}_{t|t}^{(i)}, P_{t|t}^{(i)}, \theta_t^{(i)}\}_{i=1}^N$
\end{algorithmic}
\label{alg:RBFSGD}
\end{algorithm}

\section{Theoretical Bounds on Posterior Approximation Error}
\label{sec:theory}

Once the joint posterior is approximated with RBSGD or RBFSGD algorithms, we can compute the marginal state posterior. In this section, we derive bounds on the approximation error in the marginal state posterior due to the approximate parameter posterior. 

As in the RBPF~\eqref{eq:approx_posterior},  we can obtain an approximation of the joint posterior with RBSGD and RBFSGD filters as,
\begin{align}\label{eq:approx_posterir_RBSGD}
      \tilde{p}(x_t, \theta_t|y_{1:t}) = \frac{1}{N}\sum_{i=1}^{N} \delta_{\theta^{(i)}_t}(\theta_t)\mathcal{N}(x_t;\hat{x}^{(i)}_{t|t}, P^{(i)}_{t|t}),
\end{align}
where $\{\theta^{(i)}_t \}_{i=1}^{N}$ are the particles (with equal weights $\frac{1}{N}$) deterministically transported towards the target posterior $p(\theta_t|y_{1:t})$.

From~\eqref{eq:approx_posterir_RBSGD} we can  compute the \emph{marginal} state posterior approximation by integrating out the particles,
\begin{align}
\label{eq:marginal_mixture_rbfsvgd}
\tilde p(x_t \mid y_{1:t})
&= \sum_{i=1}^{N} \tilde{p}(x_t, \theta^{(i)}_t|y_{1:t}) \nonumber \\
&= \frac{1}{N}\sum_{i=1}^N 
\mathcal{N}\!\big(x_t; \hat{x}_{t|t}^{(i)},\, P_{t|t}^{(i)}\big),
\end{align}
 by the sifting property of the Dirac delta.
\begin{remark}\label{rem:MoG}
  The representation in \eqref{eq:marginal_mixture_rbfsvgd} highlights a key
advantage of the RBSGD/RBFSGD framework: the marginal approximated state posterior is a
\emph{Gaussian mixture}, in contrast to \emph{augmented} EKF approaches, which impose a  \emph{Gaussian}
approximation on the joint state--parameter vector and thus force the marginal state
posterior to remain Gaussian.  The mixture form in
\eqref{eq:marginal_mixture_rbfsvgd} naturally captures multi-modality, skewness,
and other nonlinear effects induced by parameter uncertainty, resulting in a
more expressive and often more accurate posterior representation.  
\end{remark}

\MB{\begin{remark}[MAP estimates]
The joint posterior can also be used to compute the \emph{maximum a-posteriori} (MAP)
estimates. Since \eqref{eq:approx_posterir_RBSGD} is a mixture of Diracs in $\theta$, the MAP
parameter is one of the particles, $\theta^{\mathrm{MAP}}_t = \theta^{(i^\ast)}_t$, with
\begin{equation*}
i^\ast = \arg\max_{i\in\{1,\dots,N\}}
\mathcal{N}\big(\hat{x}^{(i)}_{t|t}; \hat{x}^{(i)}_{t|t}, P^{(i)}_{t|t}\big),
\qquad
x^{\mathrm{MAP}}_t = \hat{x}^{(i^\ast)}_{t|t},
\end{equation*}
since each conditional Gaussian is maximized at its mean $\hat{x}^{(i)}_{t|t}$.
\end{remark}}

To assess the reliability of the 
approximated posteriors with the particle methods, in this  section, we establish a bound on how the approximation error in the parameter posterior affects the resulting state estimates. The following Proposition formalizes the result by relating the KL divergence between the true and approximate parameter posteriors to the total variation distance between the corresponding marginal state posteriors.

\begin{proposition}
\label{prop:pinsker_bound}
Let $\pi_t(\theta)=p(\theta_t| y_{1:t})$ denote the true parameter posterior and let $q_t(\theta)$
be its \MB{SVGD-based} particle approximation at time $t$\MB{, assumed absolutely continuous
w.r.t.\ $\pi_t$ with $\mathrm{KL}\big(q_t \,\|\, \pi_t\big)<\infty$}. Let $p(x_t| y_{1:t})$ and
$\tilde p(x_t| y_{1:t})$ denote the true and approximate marginal state posteriors, obtained by
marginalizing the joint posterior $p(x_t,\theta_t| y_{1:t})$ over $\pi_t$ and $q_t$,
respectively. Then,
\begin{equation}
\label{eq:pinsker_tonelli_bound}
\|p(x_t| y_{1:t}) - \tilde{p}(x_t| y_{1:t})\|_{1}
\le
\sqrt{2\mathrm{KL}\big(q_t(\theta) \,\|\, \pi_t(\theta)\big)},
\end{equation}
where $\|\cdot\|_1$ denotes the $L^1$ norm \MB{with respect to the Lebesgue measure}.
\end{proposition}

\begin{proof}
By definition we have,
\begin{align*}
p(x_t| y_{1:t}) &= \int p(x_t|\theta, y_{1:t})\,\pi_t(\theta)\,d\theta, \\
\tilde{p}(x_t| y_{1:t}) &= \int p(x_t | \theta, y_{1:t})\,q_t(\theta)\,d\theta,
\end{align*}
and the $L^1$ distance between the marginals is,
\begin{equation*}
\begin{aligned}
\|p(x_t \mid y_{1:t}) &- \tilde{p}(x_t \mid y_{1:t})\|_{1}
= \\
&\int \Big| \int p(x_t \mid \theta, y_{1:t}) \,
[\pi_t(\theta) - q_t(\theta)] \, d\theta 
\Big|dx_t.
\label{eq:posterior_diff}
\end{aligned}
\end{equation*}
Applying the triangle inequality in the  integral above yields,
\begin{equation}
\begin{aligned}
\label{eq:inner_bound}
\|p(x_t \mid y_{1:t}) &- \tilde{p}(x_t \mid y_{1:t})\|_{1}
\le \\
&\int \int \big| p(x_t \mid \theta, y_{1:t}) [\pi_t(\theta)-q_t(\theta)] \big| \, d\theta \, dx_t.
\end{aligned}
\end{equation}
The integrand on the right-hand side in \eqref{eq:inner_bound} is nonnegative, thus we can exchange the order of integration:
\begin{align}\label{eq:ton}
\int \int& \big| p(x_t | \theta, y_{1:t}) [\pi_t(\theta)-q_t(\theta)] \big| \, d\theta \, dx_t \nonumber \\
&= \int \Big( \int p(x_t \mid \theta, y_{1:t}) \, dx_t \Big) \, |\pi_t(\theta)-q_t(\theta)| \, d\theta \nonumber\\
&= \int |\pi_t(\theta)-q_t(\theta)| \, d\theta,
\end{align}
as $p(x_t| \theta,y_{1:t})$ is a probability density in $x_t$ for a fixed $\theta$, it integrates to one. Combining \eqref{eq:ton} with \eqref{eq:inner_bound} yields,
\begin{equation}
\label{eq:tv_param_bound}
\|p(x_t| y_{1:t}) - \tilde{p}(x_t| y_{1:t})\|_{1}
\le \|\pi_t(\theta) - q_t(\theta)\|_{1}.
\end{equation}

Finally, Pinsker's inequality \cite{pinsker1964} relates the $L^1$ norm and the KL divergence,
\[
\|q_t(\theta) - \pi_t(\theta)\|_1 \le \sqrt{2\,\mathrm{KL}(q_t(\theta)\|\pi_t(\theta))}.
\]
Combining the above inequality with \eqref{eq:tv_param_bound}, we obtain,
\[
\|p(x_t| y_{1:t}) - \tilde{p}(x_t| y_{1:t})\|_{1}
\le \sqrt{2\,\mathrm{KL}\big(q_t(\theta) \,\|\, \pi_t(\theta)\big)}.
\]
\end{proof}

\noindent
Proposition~\ref{prop:pinsker_bound} provides a theoretical guarantee on the fidelity of the proposed SVGD-based parameter inference: if the divergence between the approximate and true parameter posteriors is small, then the corresponding marginal state estimate remains close in total variation.

\MB{\begin{remark}[Filter-agnostic bound]
Proposition~\ref{prop:pinsker_bound} applies to any Rao--Blackwellized filter
(EKF, UKF, RBPF), provided the conditional state distribution $p(x_t\mid\theta,y_{1:t})$
is a proper density and the parameter approximation $q_t$ satisfies the stated
absolute-continuity and finite-KL condition.
\end{remark}}

Combining the error bound from Proposition~\ref{prop:pinsker_bound} with Proposition~\ref{thm:svgd_one_step_formal} allows us to
establish that, under the step-size schedule of Proposition~\ref{cor:state_contraction}, the SVGD iterations
monotonically tighten the upper bound on the total variation distance between
the true and approximate marginal state posteriors.
The following proposition formalizes this result.

\begin{proposition}\label{cor:state_contraction}
    Let $\pi_t(\theta) = p(\theta_t|y_{1:t})$ denote the true parameter posterior and let
$\{q_{t,m}\}_{m\ge0}$ be the sequence of proposals with SVGD iterates produced with step sizes
\MB{$0<\epsilon_m<\epsilon_{0,m-1}$} as in Proposition~\ref{thm:svgd_one_step_formal} at time $t$.  Let
$p(x_t| y_{1:t})$ be the true marginal state posterior and let
$\tilde p_{m}(x_t| y_{1:t})$ be the approximated marginal state posterior obtained by
integrating w.r.t.\ $q_{t,m}$.  Then at the population limit ($N \to \infty$), for all $m\ge0$,
\begin{equation}
\label{eq:state_error_contraction}
\big\|
p(x_t| y_{1:t})-\tilde p_{m}(x_t|y_{1:t})
\big\|_{1} \le \sqrt{2\,\mathrm{KL}(q_{t, m} \,\|\, \pi_t)}.
\end{equation}
\MB{Moreover, since $\epsilon_m<\epsilon_{0,m-1}$, Proposition~\ref{thm:svgd_one_step_formal}
gives $\mathrm{KL}(q_{t,m}\|\pi_t)<\mathrm{KL}(q_{t,m-1}\|\pi_t)$ whenever $q_{t,m-1}\ne\pi_t$,
so the right-hand side of \eqref{eq:state_error_contraction} strictly decreases with $m$ until
$q_{t,m}=\pi_t$.}
\end{proposition}

\begin{proof}
The bound \eqref{eq:state_error_contraction} follows from applying
Proposition~\ref{prop:pinsker_bound} to $q_{t,m}$. \MB{The strict decrease follows from
Proposition~\ref{thm:svgd_one_step_formal}, since $\epsilon_m<\epsilon_{0,m-1}$ ensures
KL divergence decreases at each step.}
\end{proof}

\noindent
Proposition~\ref{cor:state_contraction} shows that the refinement of the
parameter posterior via SVGD directly improves the theoretical upper bound on
the resulting state approximation error.  Hence, \MB{with step sizes satisfying $\epsilon_m<\epsilon_{0,m-1}$}, repeated SVGD updates
jointly tighten the parameter and state posteriors \MB{monotonically until convergence}.

\begin{remark}
\label{rem:population_limit_svgd}
Although the practical SVGD implementation is based on a finite set of $N$
particles, the theoretical analysis in
Propositions~\ref{thm:svgd_one_step_formal}--\ref{prop:pinsker_bound} and
Proposition~\ref{cor:state_contraction} is carried out asymptotically, assuming
$N \rightarrow \infty$. In this regime, the particle system is interpreted as
an empirical approximation to an underlying smooth density \(q_t\), which
evolves under the deterministic transport map
\(\theta \mapsto \theta + \epsilon \phi(\theta)\). This viewpoint is
standard in SVGD theory (see, e.g.,~\cite{liu2016stein,korba2021kernel}). The
particle-based algorithm therefore serves as a finite-dimensional
discretization of the mean-field SVGD dynamics analysed here. \MB{Likewise, the
strict-decrease guarantee of Proposition~\ref{cor:state_contraction} assumes a
step size $\epsilon_m<\epsilon_{0,m-1}$ that shrinks with the iterate, whereas a
small constant step size is used in practice. Both idealizations are relaxed in
the experiments, where the finite-particle, constant-step algorithm is shown to
perform well.}
\end{remark}

\MB{\section{Constrained Parameter Estimation via SVGD-ADMM}
\label{sec:admm}
In many applications the parameters are subject to physical \emph{constraints}, such as
      box bounds, normalization, energy budgets, or equality relations. The SVGD
particle update in \eqref{eq:svgd_update} is unconstrained and cannot enforce such
requirements; naive clamping distorts the particle geometry that the Fisher-informed
kernel exploits. Since the SVGD update is itself a gradient step, it composes naturally
with the Alternating Direction Method of Multipliers (ADMM)~\cite{boyd2011admm}, which
splits a constrained problem into an unconstrained gradient (primal) step, a projection,
and a dual update.

The parameters should respect the constraints $\theta \in \mathcal{C}$ and we enforce  the constraint  $\mathcal{C}$ on each particle by running scaled ADMM inside
the SVGD update. Each particle $\theta^{(i)}_{t,m}$ carries an auxiliary variable $z^{(i)}_{t,m}$ and a scaled dual
variable $\nu^{(i)}_{t,m}$, updated per inner iteration $m$ as
\begin{subequations}
\label{eq:admm}
\begin{align}
\theta^{(i)}_{t,m} &\leftarrow \theta^{(i)}_{t,m-1} + \epsilon\,\hat\phi^\ast_\rho(\theta^{(i)}_{t,m-1}),
   \label{eq:admm_primal}\\
z^{(i)}_{t,m} &\leftarrow \Pi_{\mathcal{C}}\!\big(\theta^{(i)}_{t,m} + \nu^{(i)}_{t,m-1}\big),
   \label{eq:admm_proj}\\
\nu^{(i)}_{t,m} &\leftarrow \nu^{(i)}_{t,m-1} + \theta^{(i)}_{t,m} - z^{(i)}_{t,m},
   \label{eq:admm_dual}
\end{align}
\end{subequations}
where $\Pi_{\mathcal{C}}$ is the Euclidean projection onto $\mathcal{C}$, available in
closed form for the constraint sets considered: box, affine, and norm-ball. The primal
direction $\hat\phi^\ast_\rho$ is the SVGD direction \eqref{eq:svgd_empirical} with the
augmented-Lagrangian penalty added to the log-posterior gradient,
\begin{equation}
\label{eq:phi_rho}
\begin{split}
\hat\phi^\ast_\rho(\theta) =
&\frac{1}{N}\sum_{j=1}^{N}\Big[
k(\theta^{(j)},\theta)\big(\nabla_{\theta^{(j)}}\log p(\theta^{(j)}|y_{1:t})\\
& -\! \rho\,(\theta^{(j)} - z^{(j)} + \nu^{(j)})\big)
\!+\! \nabla_{\theta^{(j)}} k(\theta^{(j)},\theta)\Big],
\end{split}
\end{equation}
with penalty parameter $\rho>0$. Since the penalty modifies only the gradient term inside 
$\hat\phi^\ast_\rho$, the kernel, Fisher preconditioner, and Fisher--Adam step of
Sections~\ref{sec:RBSGD}--\ref{sec:RBFSGD} are reused unchanged; applying this to RBSGD and
RBFSGD yields their constrained variants, denoted \emph{cRBSGD} and \emph{cRBFSGD}. The
auxiliary $z^{(i)}$ is feasible by construction and is returned as the parameter estimate,
and the ADMM variables are warm-started across time steps, so feasibility is maintained
online at the cost of one penalty term and one projection per particle per iteration.

\begin{remark}[Scope of the theoretical guarantees]
For fixed auxiliary variables $z^{(i)},\nu^{(i)}$, the primal step \eqref{eq:admm_primal} is a
Stein update for the penalized target $\propto p(\theta\mid y_{1:t})\exp(-\tfrac{\rho}{2}\|\theta-z+\nu\|^2)$,
and thus inherits the per-step KL-descent property of
Proposition~\ref{thm:svgd_one_step_formal} with respect to that target. The consistency bound
of Proposition~\ref{prop:pinsker_bound} applies to the returned parameter distribution, as it
is agnostic to how the particles are produced. The overall convergence of the coupled primal,
projection, and dual updates to a feasible optimum is governed by standard ADMM
theory~\cite{boyd2011admm}; we do not claim the contraction of
Proposition~\ref{cor:state_contraction} for the projected iterate.
\end{remark}}

\section{Case Studies}
\label{sec:case_studies}
\MB{In this section, we evaluate the performance of the proposed filtering algorithms on three
case studies. The first two involve constrained parameter estimation. The first considers a
nonlinear fed-batch bioreactor process, highlighting the filters' capability to handle
time-varying nonlinearities and process uncertainties under a physical inequality constraint
on the estimated parameter. The second considers a damped pendulum, where the estimated
parameters are subject to a nonlinear equality constraint that a standard particle filter
cannot enforce, illustrating the advantage of the ADMM-constrained variants. In the third
case study, we consider the online training of a neural network to learn the unknown
nonlinear dynamics of a mass--spring--damper system.

The filters' performances are assessed using the \emph{Continuous Ranked Probability Score} (CRPS), which provides a scoring rule for evaluating the accuracy of probabilistic estimates. For a scalar measurement $y$ and a predicted cumulative distribution function (CDF) $F$, the CRPS is defined as
\begin{equation}
\label{eq:crps_def}
\mathrm{CRPS}(F, y) = \int_{-\infty}^{\infty} \big(F(z) - \mathbb{I}\{z \geq y\}\big)^2 \, dz,
\end{equation}
where $\mathbb{I}\{\cdot\}$ denotes the indicator function. Lower CRPS values indicate a more accurate predictive distribution.} 

\subsection{Constrained Estimation in a Fed-Batch Bioreactor}
\label{sec:cs_bioreactor}
We consider a fed-batch bioreactor with Haldane substrate inhibition and product-inhibited
growth~\cite{marcos2004adaptive}. The states $x=[X,S,P]$ are the biomass, substrate, and
product concentrations, and the input $D$ is the dilution rate. The specific growth rate is
\begin{equation}
\label{eq:bioreactor_mu}
\mu(S,P,\eta_t)=\eta_t\,\frac{\mu_{\max}S}{K_s+S+S^2/K_i}\Big(1-\frac{P}{P_{\max}}\Big)_{\!+},
\end{equation}
where $(a)_{+}=\max(a,0)$, giving the dynamics
\begin{equation}
\label{eq:bioreactor_ode}
\begin{aligned}
\dot X &= \mu X - k_{dp}\,P X - D X,\\
\dot S &= -\tfrac{1}{Y_{xs}}\,\mu X + D\,(S_f - S),\\
\dot P &= (Y_{px}\,\mu + \beta)\,X - D P,
\end{aligned}
\end{equation}
with output $y=P$, where $\mu$ abbreviates $\mu(S,P,\eta_t)$. The mixing efficiency $\eta_t$ (which enters the dynamics only through $\mu$) is
time-varying and is the parameter to be estimated. 
The remaining constants are  known, with values given in
Table~\ref{tab:bioreactor_params}. The dynamics are discretized with an RK4 scheme at sampling time $t_s=0.5$~h. The process-noise covariance is fixed at $Q_t=10^{-6}I_3$, while the measurement-noise covariance is estimated online, with ground-true value $R_t=10^{-3}$.

\begin{table}[!t]
\centering
\caption{Fed-batch bioreactor model parameters.}
\label{tab:bioreactor_params}
\setlength{\tabcolsep}{4pt}
\begin{tabular}{llll@{\hskip 1em}llll}
\hline
Param. & Value & Unit & & Param. & Value & Unit \\
\hline
$\mu_{\max}$ & $0.4$  & h$^{-1}$ & & $Y_{px}$  & $0.6$  & g/g \\
$K_s$        & $0.1$  & g/L      & & $\beta$   & $0.05$ & h$^{-1}$ \\
$K_i$        & $10.0$ & g/L      & & $k_{dp}$  & $0.02$ & L\,g$^{-1}$h$^{-1}$ \\
$P_{\max}$   & $0.6$  & g/L      & & $Y_{xs}$  & $0.5$  & g/g \\
$S_f$        & $5.0$  & g/L      & &           &        & \\
\hline
\end{tabular}
\end{table}
The objective is to jointly estimate the states $x$, the measurement-noise covariance $R_t$,
and the mixing efficiency $\eta_t$, subject to the physical constraint $\eta_t\in[0.05,1.0]$.
This inequality is enforced by the ADMM projection of Section~\ref{sec:admm}, giving the
constrained filters cRBFSGD and cRBPF. The proposed filters represent the parameter posterior
with $N=10$ particles, each updated with $M=1$ Fisher--Adam step per instant at step size
$\epsilon=5\times10^{-3}$, while the cRBPF uses $N=40$ particles. An EKF with fixed $\eta=1$
serves as a baseline, and all filters start from the same initial estimates.

Figure~\ref{fig:bioreactor_state} shows the state estimation results. The constrained
RB filters track $X$ and $S$ throughout the run, whereas the EKF, which cannot adapt its
fixed mixing efficiency to the drifting $\eta_t$, progressively diverges once the parameter
departs from its initial value. Figure~\ref{fig:bioreactor_eta} reports the parameter
estimates. The proposed cRBFSGD filter tracks the time-varying $\eta_t$ with substantially
tighter uncertainty than the constrained particle filter cRBPF. Moreover, it recovers the true measurement-noise covariance $R_t$, while
the cRBPF estimate settles below the true value.

To assess robustness, Fig.~\ref{fig:bioreactor_mc} reports a Monte Carlo study over
$20$ independent realizations of $\eta_t$, comparing cRBFSGD,
cRBSGD, and cRBPF in terms of CRPS on the states and on the parameters. Both Stein-based
filters achieve markedly lower and less variable state CRPS than the cRBPF. On the parameter
CRPS, the Fisher-preconditioned cRBFSGD attains the lowest median, followed by cRBSGD and
then cRBPF, confirming the benefit of the Fisher preconditioning. Notably, the proposed filters achieve this using $N=10$ particles,
four times fewer than the $N=40$ used by the cRBPF.

\begin{figure}[!t]
\centering
\includegraphics[width=\columnwidth]{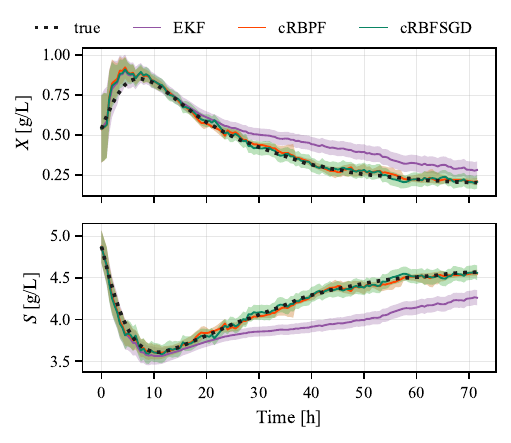}
\caption{Bioreactor state estimation: true states (dotted), filter means (solid), and $95\%$ confidence intervals (shaded).}
\label{fig:bioreactor_state}
\end{figure}

\begin{figure}[!t]
\centering
\includegraphics[width=\columnwidth]{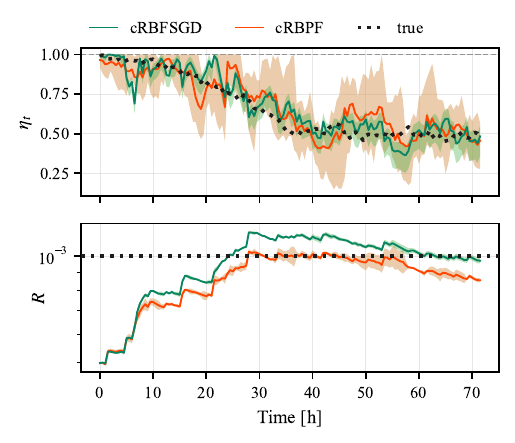}
\caption{Mixing efficiency $\eta_t$ estimation (top) and measurement-noise covariance $R_t$ estimation (bottom): true parameters (dotted), filter means (solid), and area between particles (shaded).}
\label{fig:bioreactor_eta}
\end{figure}

\begin{figure}[!t]
\centering
\includegraphics[width=\columnwidth]{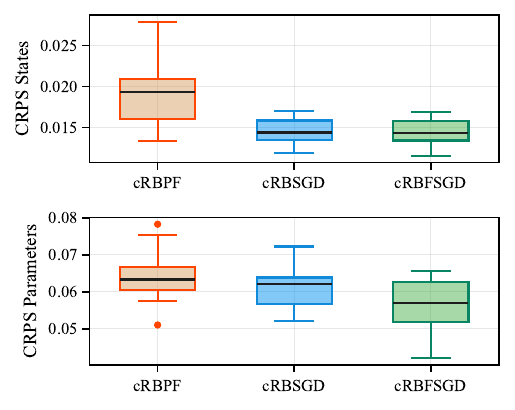}
\caption{Monte Carlo CRPS on states (top) and parameters (bottom) over
$20$ runs. cRBFSGD attains the lowest parameter CRPS.}
\label{fig:bioreactor_mc}
\end{figure}

\subsection{Constrained Estimation for a Damped Pendulum}
\label{sec:cs_pendulum}
We consider a damped pendulum consisting of a point mass $m$ at the end of a rigid rod of
length $L$, with gravitational acceleration $g$, linear damping coefficient $b$, and applied
torque $\tau$. The states $x=[\varphi,\psi]$ are the angle from vertical and the angular
velocity, governed by
\begin{equation}
\label{eq:pendulum_ode}
\dot\varphi=\psi,\qquad
\dot\psi=-\frac{g}{L}\sin\varphi-\frac{b}{mL^2}\,\psi+\frac{\tau}{mL^2},
\end{equation}
with angle-only output $y=\varphi$. The rod length and mass $\theta=[L,m]$ are estimated
from the angle measurements, with $g=9.81$ and $b=0.15$ fixed and known.

Here the objective is to estimate the states together with the length $L$ and mass $m$,
subject to the known moment of inertia $mL^2=J=0.5$. In log-coordinates this nonlinear
constraint becomes affine, $2\log L+\log m=\log J$, which the ADMM projection of
Section~\ref{sec:admm} enforces exactly, yielding the constrained filter cRBFSGD. For a
constraint of this form there is no straightforward way to restrict the RBPF particles, so it
serves as an unconstrained baseline. The proposed filter uses $N=10$ particles with $M=1$
Fisher--Adam step per instant and step size $\epsilon=10^{-2}$, while the RBPF uses $N=50$. Both are discretized with an RK4 scheme
at sampling time $t_s=0.1$~s, and the system is driven by the free decay of a large initial
angle ($\varphi_0=0.8$~rad).

Figure~\ref{fig:pendulum_state} shows the state estimates. Both filters track the angle
$\varphi$ and the unobserved angular velocity $\psi$ accurately, so the constraint does not
degrade state tracking. The difference appears in the parameter estimates of
Fig.~\ref{fig:pendulum_theta}. The cRBFSGD filter keeps $mL^2$ exactly at its true value and
its estimates of $L$ and $m$ converge to the truth. The RBPF cannot enforce the constraint as
its estimates of $L$ and $m$ are biased with wide uncertainty, and the product $mL^2$ drifts
around the true value.

\begin{figure}[t]
\centering
\includegraphics[width=\columnwidth]{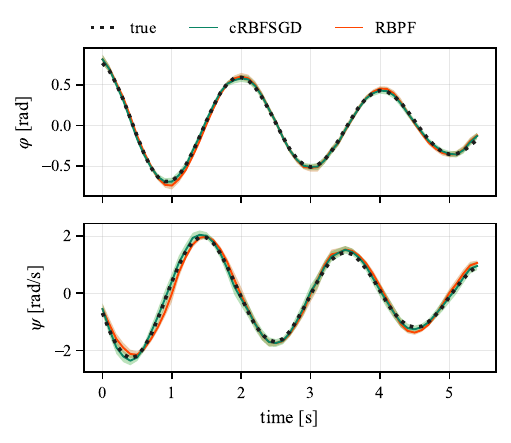}
\caption{Pendulum state estimation:  true states (dotted), filter means (solid), and $95\%$ confidence intervals (shaded).}
\label{fig:pendulum_state}
\end{figure}

\begin{figure}[t]
\centering
\includegraphics[width=\columnwidth]{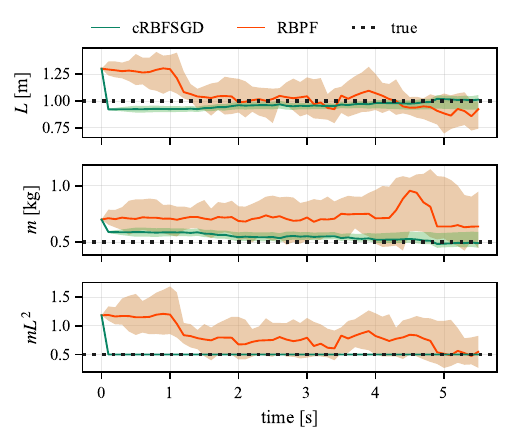}
\caption{Pendulum parameter estimation: length $L$, mass $m$, and inertia $mL^2$. cRBFSGD
satisfies the constraint $mL^2=J$ exactly, while the RBPF estimate wanders around the true value.}
\label{fig:pendulum_theta}
\end{figure}

\subsection{Online Neural-Network Learning of Unknown Dynamics}
\label{sec:cs_nn}
We consider a mass--spring--damper system with an unknown nonlinear
restoring force,
\begin{equation}
\label{eq:nn_ode}
\dot x_1=x_2,\qquad
\dot x_2=-\frac{k}{M}x_1-\frac{c}{M}x_2-\frac{1}{M}g(x_1)+\frac{u}{M},
\end{equation}
where $x=[x_1,x_2]$ are position and velocity, $u$ is the applied force, and the output is the
position $y=x_1$. The mass $M$, stiffness $k$, and damping $c$ are known, whereas the
restoring-force correction $g(x_1)$ is unknown and is approximated online by a neural network
$g(x_1)\approx f_{\mathrm{nn}}(x;\theta)$. The network is a two-hidden-layer perceptron with
$\tanh$ activations ($73$ weights). Normalizing the input and output data can significantly affect the accuracy and stability of
neural-network training. Since the statistical properties of the states are not known a
priori, we treat the normalization statistics as additional learnable parameters, adapted online, as proposed
in~\cite{DEHKORDI2026112659}. This yields a total number of parameters
$\theta\in\mathbb{R}^{75}$, comprising the $73$ network weights and $2$ normalization parameters.

The objective is to jointly estimate the states $x$,
and the network parameters $\theta$ from noisy position measurements. The parameters are
estimated online by the proposed RBFSGD filter, with the unknown term $f_{\mathrm{nn}}$ treated
as part of the process model. An RBPF using the same augmented model is included for comparison.
Both filters use $N=20$ particles with $M=1$ Fisher--Adam step per instant and step size
$\epsilon=10^{-2}$, are discretized with an RK4 scheme at sampling time $t_s=0.1$~s, and are
initialized identically. The system is driven by the free decay of the initial condition
$x_0=[1.5,\,0]$, with the true correction $g(x_1)=10\,x_1\exp(-x_1^2/2)$.

Figure~\ref{fig:nn_oscillator} shows the results. The RBFSGD filter tracks the position $x_1$
and the unobserved velocity $x_2$ accurately over the entire horizon, and its online estimate
of the correction $g$ learns to recover the shape of the true nonlinearity as the trajectory explores
the state space. The RBPF, by contrast, fails to learn the correction adequately and its state
estimate diverges from the true trajectory, losing phase after the first oscillations. The gradient-based SVGD update pairs naturally with gradient-driven neural-network training,
while the Rao--Blackwellized structure confines the sampling to the parameters. The RBPF,
whose updates use no gradient information, fails to train the network effectively here.

\begin{figure}[t]
\centering
\includegraphics[width=\columnwidth]{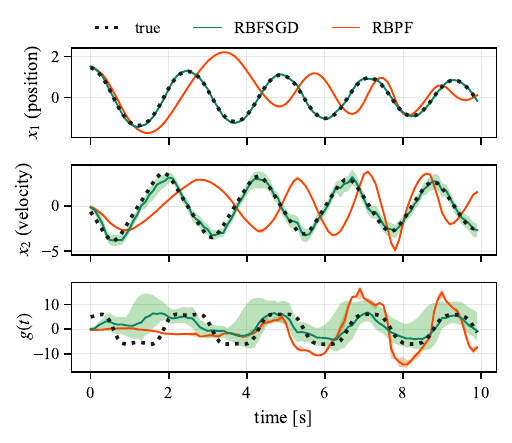}
\caption{Online NN learning: true (dotted), RBFSGD (green), and RBPF (orange). RBFSGD tracks
position $x_1$, velocity $x_2$, and the unknown correction $g(t)$, while the RBPF loses phase.}
\label{fig:nn_oscillator}
\end{figure}

\section{Conclusion}
This work developed Rao–Blackwellized filters in which conditional state distributions are computed by Kalman filtering, while parameter particles are deterministically updated through SVGD. Fisher-informed preconditioning was introduced to account for parameter scaling and correlations, and an ADMM-based formulation was used to enforce physical constraints. The theoretical analysis established, at the population level, a bound relating the parameter-posterior approximation error to the resulting error in the marginal state posterior.

The numerical results highlighted different benefits of the proposed framework. In the bioreactor example, the Stein-based filters achieved lower state CRPS than the RBPF using fewer particles, while Fisher-informed preconditioning provided the most accurate parameter estimates. In the pendulum example, the constrained formulation enforced the physical relation exactly without compromising state tracking. In the neural-network example, RBFSGD learned the unknown nonlinear correction online and maintained accurate tracking, whereas the RBPF baseline lost phase. These results indicate that the proposed approach is particularly useful when the particle budget is limited and gradient information or physical constraints can be exploited. Future research will investigate the integration of the estimated
state--parameter distributions into probabilistic and robust control schemes.

\bibliographystyle{plain}
\bibliography{references}
\appendix 
\section{Proof of Proposition~\ref{thm:svgd_one_step_formal}}\label{proof:svgd_one_step_formal}
The proof is based on the  SVGD contraction argument from  \cite[Theorem 3.1,  Lemma 3.2]{liu2016stein}. 

For brevity, in this proof, we drop the time index $t$ and only consider the iteration index $m$.  Recall that  $q_m$ is the distribution associated with updated parameters $\theta_m^{(i)}$ transported via  $\theta_{m}^{(i)} =\theta_{m-1}^{(i)} + \epsilon \phi^{\ast}_{q_{m-1},p}(\theta_{m-1}^{(i)})$. Thus, for a given set of parameters $\{\theta_{m-1}^{(i)}\}_{i=1}^N$, we can view $q_m$ as a function of $\epsilon$ and let us denote it as $q_m({\epsilon})$. Consequently, we note that the proposal at the previous step is $q_{m-1}= q_{m}(0)$.
A second-order Taylor expansion of the KL divergence under  smoothness assumptions yields,
\begin{align}
\label{eq:taylor_expansion}
\mathrm{KL}(q_{m}(\epsilon)\|\pi)
=&
\mathrm{KL}(q_{m-1}\|\pi)
+
\epsilon \left.\frac{d}{d\epsilon}
\mathrm{KL}(q_{m}(\epsilon)\|\pi)\right|_{\epsilon=0} \notag \\
&\quad +
\frac{\epsilon^2}{2}
\left.\frac{d^2}{d\epsilon^2}
\mathrm{KL}(q_{m}(\epsilon)\|\pi)\right|_{\epsilon=\xi},
\end{align}
for some $0 \leq \xi \leq \epsilon$.

Following  \cite[Lemma 3.2]{liu2016stein}, at the population limit ($N \to \infty$),  the first derivative equals
\begin{align}
\label{eq:D_lemma}
\left.\frac{d}{d\epsilon}
\mathrm{KL}(q_{m}(\epsilon)\|\pi)\right|_{\epsilon=0}
=
-\mathbb{D}(q_{m-1}\|\pi)^2,
\end{align}
where $\mathbb{D}(q_{m-1}\|\pi)^2$ is the Kernel Stein Discrepancy. \MB{Since the kernel is
$C_0$-universal, $\mathbb{D}(q_{m-1}\|\pi)>0$ whenever $q_{m-1}\ne\pi$~\cite{liu2016stein}, so} we can define
$B_{m-1} > 0$ as the following:
\begin{equation}
\label{eq:B_def}
    B_{m-1} \vcentcolon= \mathbb{D}(q_{m-1}\|\pi)^2
\end{equation}
Boundedness of the Hessian of $\log\pi$
and smoothness of the kernel $k$ ensure the second derivative in \eqref{eq:taylor_expansion} is uniformly bounded
by a constant $C < \infty$. Thus via \eqref{eq:taylor_expansion}, \eqref{eq:D_lemma}, and \eqref{eq:B_def}, we have
\begin{align}
\label{eq:dKL_ineq}
\mathrm{KL}(q_{m}(\epsilon)\|\pi)
\le
\mathrm{KL}(q_{m-1}\|\pi)
-\epsilon B_{m-1} + \frac{\epsilon^2}{2} C,
\end{align}

Finally, by choosing $\epsilon < \MB{\epsilon_{0,m-1}} = \frac{2B_{m-1}}{C}$, we obtain
\eqref{eq:KL_statement}. \MB{Note that $\epsilon_{0,m-1}=2\,\mathbb{D}(q_{m-1}\|\pi)^2/C$
depends on the current iterate and vanishes as $q_{m-1}\to\pi$.} This completes the proof.

\end{document}